\documentclass[a4paper,11pt]{article}
\usepackage{algorithm}
\usepackage{algorithmic}
\usepackage{mathrsfs}
\usepackage{amsmath,amsthm,amssymb}
\usepackage{cite}
\usepackage{txfonts}
\usepackage{graphics}
\usepackage{float}
\usepackage{epstopdf}
\usepackage{epsfig}
\usepackage{tikz}
\usepackage{caption}

\usepackage[T1]{fontenc}
\usepackage[utf8]{inputenc}
\usepackage{authblk}

\newtheorem{definition}{Definition}

\newtheorem{theorem}{Theorem}

\newtheorem{lemma}{Lemma}
\newtheorem{Claim}{Claim}
\newtheorem{Conjecture}{Conjecture}
\title{On the Sensitivity Complexity of $k$-Uniform Hypergraph Properties}

\author[1]{Qian Li}
\author[1]{Xiaoming Sun}
\affil[1]{Institute of Computing Technology, Chinese Academy of Sciences, Beijing, China

  \texttt{\{liqian, sunxiaoming\}@ict.ac.cn}}
\date{\vspace{0em}}
\begin{document}
\maketitle
\begin{abstract}
In this paper we investigate the sensitivity complexity of hypergraph properties. We present a $k$-uniform hypergraph property with sensitivity complexity $O(n^{\lceil k/3\rceil})$  for any $k\geq3$, where $n$ is the number of vertices. Moreover, we can do better when $k\equiv1$ (mod 3) by presenting a $k$-uniform hypergraph property with sensitivity $O(n^{\lceil k/3\rceil-1/2})$. This result disproves a conjecture of Babai~\cite{Babai}, which conjectures that the sensitivity complexity of $k$-uniform hypergraph properties is at least $\Omega(n^{k/2})$. We also investigate the sensitivity complexity of other symmetric functions and show that for many classes of transitive Boolean functions the minimum achievable sensitivity complexity can be $O(N^{1/3})$, where $N$ is the number of variables. Finally, we give a lower bound for sensitivity of $k$-uniform hypergraph properties, which implies the {\em sensitivity conjecture} of $k$-uniform hypergraph properties for any constant $k$.
\end{abstract}
\noindent

\section{Introduction}
In order to understand the effect of symmetry on computational complexity, especially in the decision tree model, Boolean functions with certain symmetry have been extensively investigated. It is observed that symmetry usually implies high complexity
or makes the problem harder in the decision tree model. An illustrative example is the well known {\em evasiveness conjecture}, which asserts that any monotone transitive Boolean function is evasive, and it has attracted a lot of attention~\cite{Evasiveness,CKS01,Kulkarni13,Stacs10}. Rivest and Vuillemin~\cite{RivestV76} showed that any non-constant monotone graph property are weekly evasive. Kulkarni et al.~\cite{KQS15} showed the analogous result for 3-hypergraph properties. Black~\cite{Black2015} extended these results to $k$-uniform hypergraph properties for any fixed $k$.

Sensitivity complexity is an important complexity measure of a Boolean function in the decision tree model, and sensitivity complexity of Boolean functions with certain symmetry has also attracted a lot of attention. One of the most challenging problem here is whether {\em symmetry} implies high sensitivity complexity. The famous sensitivity conjecture, which asserts sensitivity complexity and block sensitivity are polynomially related, implies $s(f)=\Omega(n^{\alpha})$ for transitive functions with some constant $\alpha>0$ since it has been shown that $bs(f)=\Omega(n^{1/3})$ for transitive functions~\cite{Sun07}. Turan~\cite{turan} initiated the study of sensitivity of graph properties and proved that the sensitivity is greater than $n/4$ for any nontrivial graph property, where $n$ is the number of vertices, and this relation is also tight up to a constant factor. He also pointed out that for symmetric functions, $s(f)\geq n/2\geq bs(f)/2$.  Recently Sun improved the lower bound to $\frac{6}{17}n$~\cite{Sun}, and Gao et al.~\cite{GMSZ13} investigated the sensitivity of bipartite graph properties as well.
In 2005, Chakraborty~\cite{Chakraborty} constructed a minterm cyclically invariant Boolean function whose sensitivity is $\Theta(n^{1/3})$, which answers Turan's question~\cite{turan} in the negative.
He also showed this bound is tight for minterm transitive functions.

For hypergraph properties, Biderman et al.~\cite{Babai} present a sequence of $k$-uniform hypergraph properties with sensitivity $\Theta(\sqrt{N})$, where $N={n\choose k}$ is the number of variables. Babai conjectures that this bound is tight, i.e., $s(f)=\Omega(\sqrt{N})$ for any nontrivial $k$-uniform hypergraph property $f$.\\

\noindent\textbf{Our Results.} In this paper we disprove this conjecture by constructing  $k$-uniform hypergraph properties with sensitivity $O(n^{\lceil k/3 \rceil})$, i.e.,
\begin{theorem}\label{upperbound:general}
For any $k\geq 3$, there exist a sequence of $k$-uniform hypergraph properties $f$ such that $s(f)=O(n^{\lceil k/3 \rceil})$.
\end{theorem}
Moreover, we can give  better constructions when $k\equiv1$ (mod 3).
\begin{theorem}\label{upperbound:3k+1}
For any $k\geq 4$ satisfying $k\equiv1$ (\emph{mod 3}) , there exist a sequence of $k$-uniform hypergraph properties $f$ such that $s(f)=O(n^{\lceil k/3 \rceil-1/2})$.
\end{theorem}
More generally, we also investigate the sensitivity of $k$-partite $k$-uniform hypergraph properties. Actually, the constructions of $k$-uniform hypergraph properties are inspired by the constructions of $k$-partite $k$-uniform hypergraph properties.
\begin{theorem}\label{thm:partite1}
For any $k\geq 3$, there exist a sequence of $k$-partite $k$-uniform hypergraph properties $f:\{0,1\}^{n^k}\rightarrow\{0,1\}$ such that $s(f)=O(n^{\lceil k/3 \rceil})$.
\end{theorem}
\begin{theorem}\label{thm:partite2}
For any $k\geq 4$ satisfying $k\equiv1$ (\emph{mod 3}) , there exist a sequence of $k$-partite $k$-uniform hypergraph properties $f:\{0,1\}^{n^k}\rightarrow\{0,1\}$ such that $s(f)=O(n^{\lceil k/3 \rceil-1/2})$.
\end{theorem}

Let $G$ be an Abelian group, the fundamental theorem of finite abelian groups states that $G\cong C_{m_1}\times\cdots\times C_{m_l}$, where $C_m$ is the cyclic group of order $m$ and $|G|=\prod_{i=1}^{l}m_i$.
\begin{theorem}\label{thm:abelian}
Let $G\leq S_n$ be a transitive Abelian group, then  there exists a Boolean function $f:\{0,1\}^n\rightarrow\{0,1\}$ invariant under $G$ such that $s(f)\leq \alpha n^{1/3}$, where $\alpha$ is a number only depending on $l$.
\end{theorem}

On the other side, we prove a lower bound of the sensitivity of $k$-uniform hypergraph properties, which implies the sensitivity conjecture of $k$-uniform hypergraph properties.
\begin{theorem}\label{lowerbound}
For any constant $k$ and any non-trivial $k$-uniform hypergraph property $f$,  $s(f)=\Omega(n)$.
\end{theorem}
 Similar lower bound holds for the sensitivity of $k$-partite $k$-uniform hypergraph properties.
\begin{theorem}\label{lowerbound2}
For any constant $k$ and any non-trivial $k$-partite $k$-uniform hypergraph property $f$,  $s(f)=\Omega(n)$, where $n$ is the number of vertices in one partition.
\end{theorem}
The proof of this theorem is very similar with the proof of Theorem 2 in~\cite{GMSZ13}, except that we divide $k$-partitions into two sets of size $1$ and $k-1$ respectively first. We omit the proof in this paper.\\

\noindent\textbf{Related Work.}
Sensitivity complexity and block sensitivity are first introduced by Cook, Dwork and Reischuk~\cite{Cook,CDR86}
and Nisan~\cite{nisan1991crew} respectively, to study the time complexity of CREW-PRAMs.
Block sensitivity has been shown to be polynomially related to a number of other complexity measures~\cite{buhrman}, such as  decision tree complexity, certificate complexity,  polynomial degree and quantum query complexity, etc, except sensitivity. The famous sensitivity conjecture, proposed by Nisan and Szegedy~\cite{NS92}, asserts that block sensitivity and sensitivity complexity are also polynomially related. On one side, it is easy to see $s(f)\leq bs(f)$ for any Boolean function $f$ according to the definitions. On the other side, it is much more challenging to prove or disprove block sensitivity is polynomially bounded by sensitivity.
Despite of a lot of effort, the best known upper bound is exponential: $bs(f)\leq\max\{2^{s(f)-1}(s(f)-\frac{1}{3}),s(f)\}$~\cite{APV16}. Recently, He, Li and Sun further improve the upper bound to $(\frac{8}{9} + o(1))s(f)2^{s(f) - 1}$~\cite{Own}. The best known separation between sensitivity and block sensitivity is quadratic~\cite{AS11}: there exist a sequence of Boolean functions $f$ with $bs(f)=\frac{2}{3}s(f)^2-\frac{1}{3}s(f)$.
For an excellent survey on the sensitivity conjecture, see \cite{HKP11}. For other recent progress, see \cite{ Bop12,AP14,ABG+14,AV15,GKS15,Sze15,GNS+16,GSTW16,Tal16,Sha16,LZ16,ST16}.\\


\noindent\textbf{Organization.} We present some preliminaries in Section~\ref{section:preliminaries}, and give the proofs of Theorem~\ref{upperbound:general} and Theorem~\ref{upperbound:3k+1} in Section~\ref{section:proof_1}. We give the constructions of $k$-partite $k$-uniform hypergraph properties (Theorem~\ref{thm:partite1} and \ref{thm:partite2}) and the proof of Theorem~\ref{thm:abelian} in Section~\ref{section:general} and give the proof of Theorem~\ref{lowerbound} in Section~\ref{section:proof_2}. Finally, we conclude this paper with some open problems in Section~\ref{section:conclusion}.

\section{Preliminaries}\label{section:preliminaries}
Let $f:\{0,1\}^n\rightarrow\{0,1\}$ be a Boolean function and $[n]=\{1,2,\cdots,n\}$. For an input $x\in\{0,1\}^n$ and a subset $B\subseteq[n]$, $x^B$ denotes the input obtained by flipping all the bit $x_j$ such that $j\in B$.
\begin{definition}
The {\em sensitivity}  of $f$ on input $x$ is defined as $s(f,x):=|\{i|f(x)\neq f(x^{\{i\}})\}|$. The sensitivity, 0-sensitivity and 1-sensitivity of the function $f$ are defined as $s(f):=max_xs(f,x)$, $s_0(f)=max_{x\in f^{-1}(0)}s(f,x)$ and $s_1(f)=max_{x\in f^{-1}(1)}s(f,x)$ respectively.
\end{definition}
\begin{definition}
The {\em block sensitivity}  $bs(f,x)$ of $f$ on input $x$ is the maximum number of disjoint subsets $B_1,B_2,\cdots,B_r$ of $[n]$ such that for all $j\in[r]$, $f(x)\neq f(x^{B_j})$. The block sensitivity of $f$ is defined as $bs(f)=max_x bs(f,x)$.
\end{definition}
\begin{definition}
A {\em partial assignment} is a function $p:[n]\rightarrow\{0,1,\star\}$. We call $S=\{i|p_i\neq\star\}$ the support of this partial assignment. We define the size of $p$ denoted by $|p|$ to be $|S|$. We call x  a (full) assignment if $x:[n]\rightarrow\{0,1\}$. We say $x$ is consistent with $p$ if $x|_{S}=p$, i.e., $x_i=p_i$ for all $i\in S$.\footnote{ The function $p$ can be viewed as a vector, and  we sometimes use $p_i$ to represent $p(i)$.}
\end{definition}
\begin{definition}
For $b\in\{0,1\}$, a {\em $b-$certificate} for $f$ is a partial assignment $p$ such that $f(x)=b$ whenever $x$ is consistent with $p$.

The {\em certificate complexity} $C(f,x)$ of $f$ on input $x$ is the minimum size of $f(x)$-certificate that is consistent with $x$. The certificate complexity of $f$ is $C(f)=\max_x C(f,x)$.

 The {\em 1-certificate complexity} of $f$ is $C_{1}(f)=\max_{x\in f^{-1}(1)} C(f,x)$, and similarly we define $C_{0}(f)$.
\end{definition}
According to the definitions, it's easy to see $s(f)\leq bs(f)\leq C(f)$, $s_0(f)\leq C_0(f)$ and $s_1(f)\leq C_1(f)$.
\begin{definition}
Let $p$ and $p'$ be two partial assignments, the distance between $p$ and $p'$ is defined as $dist(p,p')=|\{i|p_i=1$ and $p'_i=0$, or $p_i=0$ and $p'_i=1\}|$.
\end{definition}
\begin{definition}
Let $f:\{0,1\}^n\rightarrow\{0,1\}$ be a Boolean function and G be a subgroup of $S_n$, we say that $f$ is invariant under $G$ if $f(x_1,\cdots,x_n)=f(x_{\sigma(1)},\cdots,x_{\sigma(n)})$ for any $x\in\{0,1\}^n$ and any $\sigma\in G$.

A Boolean function $f$ is called {\em transitive} (or {\em weakly symmetric}) if $G$ is a transitive group\footnote{A group $G\leq S_n$ is transitive if for every $i<j$, there exists a $\sigma\in G$ such that $\sigma(i)=j$.}. A Boolean function $f$ is called symetric if $G=S_n$.
\end{definition}
\begin{definition}
A transitive Boolean function $f$ is called minterm-transitive if there exist a partial assignment $p$ such that $f(x)=1$ if and only if $x$ is consistent with $p^{\sigma}:=(p_{\sigma(1)},p_{\sigma(2)},\cdots,p_{\sigma(n)})$ for some $\sigma\in G$. We call $p$ the minterm.
\end{definition}
A Boolean string can represent a graph in the following manner:  $x_{(i,j)}=1$ means there is an edge connecting vertex $i$ and vertex $j$, and $x_{i,j}=0$ means there is no such edge. Graph properties are functions which are independent with the labeling of vertices, i.e. two isomorphic graphs have the same function value.
\begin{definition}
A Boolean function $f:\{0,1\}^{n\choose 2}\rightarrow\{0,1\}$ is called a graph property if for every input $x=(x_{(1,2)},\cdots,x_{(n-1,n)})$ and every permutation $\sigma\in S_n$,
$$f(x_{(1,2)},\cdots,x_{(n-1,n)})=f(x_{(\sigma(1),\sigma(2))},\cdots,x_{(\sigma(n-1),\sigma(n))}).$$
\end{definition}

Similarly, we define $k$-uniform hypergraph properties.
\begin{definition}\label{def:hypergraph}
A Boolean function $f:\{0,1\}^{n\choose k}\rightarrow\{0,1\}$ is called a $k$-uniform hypergraph property if for every input $x=(x_{(1,2,\ldots,k)},\cdots,x_{(n-k+1,\ldots,n-1,n)})$ and every permutation $\sigma\in S_n$,
$$f(x_{(1,2,\ldots,k)},\cdots,x_{(n-k+1,\ldots,n-1,n)})=f(x_{(\sigma(1),\sigma(2),\ldots,\sigma(k))},\cdots,x_{(\sigma(n-k+1),\ldots,\sigma(n-1),\sigma(n))}).$$
Let $p$ be a partial assignment and $\sigma\in S_n$, we define $\sigma(p)$ as $\sigma(p)_S=p_{\sigma(S)}$ where $S$ is any subset of [n] of size $k$ and $\sigma(S)=\{\sigma(i)|i\in S\}$.
\end{definition}

\begin{definition}
A Boolean function $f:\{0,1\}^{n^k}\rightarrow\{0,1\}$ is called $k$-partite $k$-uniform hypergraph property, if for every input  $x=(x_{(1,1,\cdots,1)},\cdots,x_{(n,n,\cdots,n)})$ and every $\sigma=(\sigma_1,\cdots,\sigma_k)\in S_n^{\otimes k}$,
$$f(x_{(1,1,\cdots,1)},\cdots,x_{(n,n,\cdots,n)})=f(x_{(\sigma_1(1),\cdots,\sigma_k(1))},\cdots,x_{(\sigma_1(n),\cdots,\sigma_k(n))}).$$
\end{definition}
It is easy to see that any ($k$-partite) $k$-uniform hypergraph property is transitive.

\section{$k$-Uniform Hypergraph Properties}\label{section:proof_1}
In this section, we give the proofs of Theorem~\ref{upperbound:general} and Theorem~\ref{upperbound:3k+1}.

\begin{proof} ({\bf Proof of Theorem~\ref{upperbound:general}}) The function we construct is a minterm function.
Let $p$ be the minterm defining $f$, and it is constructed as follow:

First, let $k_1$ and $k_2$ be two integers such that $k_1+2k_2=k$ and $k_1,k_2\leq \lceil k/3\rceil$. Let $V=\{v_1,\cdots,v_n\}$ be the set of vertices and $B=\{v_n,v_{n-1},\cdots,v_{n-k_1+1}\}$. For each $1\leq i\leq6$, let $W_i=\{v_{(i-1)k_2+1},\cdots,v_{ik_2}\}$, and $C=\bigcup_{1\leq i\leq 6}W_i$, $D=V\setminus(C\cup B)$.
\begin{itemize}
\item For any $S\subseteq C$ of size $2k_2$,  $p(B\cup S)=0$, except $S=W_i\cup W_{i+1}$ for $i\in[5]$ where $p(B\cup S)=1$.
\item For any $S$ of size $2k_2$ and $k_2\leq|S\cap C|<2k_2$, $p(B\cup S)=1$, except $W_3$ or $W_4\subseteq S$ where $p(B\cup S)=0$.
\item All the other variables are $\star$.
\end{itemize}

If $f(x)=1$ then $x$ is consistent with some $\sigma(p)$, which implies $C(f,x)\leq|p|$. Thus $s_1(f)\leq C_1(f)\leq |p|=\sum_{i=k_2}^{2k_2}{6k_2 \choose i}{n-6k_2-k_1 \choose 2k_2-i}=O(n^{k_2})$.
 Moreover, if $f(x)=0$ then $s(f,x)$ is at most the number of isomorphisms of $p$ (i.e., $\sigma(p)$s) adjacent to $x$, thus according to the triangle inequality, $s_0(f)$ is at most the maximum number of $\sigma(p)$s where the distance between any two of them is at most $2$. We claim that for any $\pi(p)$, there are  $O(1)$ isomorphisms $\sigma(p)$s satisfying $\pi(B)=\sigma(B)$ and $dist(\pi(p),\sigma(p))\leq 2$. It is easy to see that this claim implies $s_0(f)=O(n^{k_1})$ since there are ${n\choose k_1}=O(n^{k_1})$ possible choices of the  $\sigma(B)$s, and this will end the whole proof.

 \begin{Claim}
For any  $\pi(p)$, there are  only $O(1)$ $\sigma(p)$s  satisfying $\pi(B)=\sigma(B)$ and $dist(\pi(p),\sigma(p))\leq 2$.
\end{Claim}
\noindent\emph{Proof.}
It is easy to see that this claim is equivalent to show $|\{\sigma(p)|dist(p,\sigma(p))\leq 2$ and $\sigma(B)=B\}|=O(1)$. The case for $k_2=1$ is a little special, and we discuss this case first.
\subparagraph{Case for $k_2=1$}

\begin{figure}\centering
\begin{tikzpicture}
    \tikzstyle{every node}=[draw,circle,fill=white,minimum size=4pt,
                            inner sep=0pt]

  \draw (0,0) node (W-1) [label=$W_1$]{}
    -- ++(0:1cm) node (W-2) [label=$W_2$]{}
    -- ++(0:1cm) node (W-3) [label=$W_3$]{}
    -- ++(0:1cm) node (W-4) [label=$W_4$]{}
    -- ++(0:1cm) node (W-5) [label=$W_5$]{}
    -- ++(0:1cm) node (W-6) [label=$W_6$]{};

  \draw (0,-2) node (B-1) {};
  \draw (1,-2) node (B-2) {};
  \draw (2,-2) node (B-3) {};
  \node[draw=none] at (3,-2) {$\cdots$};
  \draw (4,-2) node (B-5) {};
  \draw (5,-2) node (B-6) {};

  \foreach \x in {1,2,5,6}
    \foreach \y in {1,2,3,5,6}
        \draw (W-\x)--(B-\y);

  \draw (2.5,-2) ellipse (3 and 0.6);
  \node[draw=none] at (3,-2.3) {\tiny $D$};

\end{tikzpicture}
\caption{The graph to illustrate $p$ for $k_2=1$}\label{figure}
\end{figure}
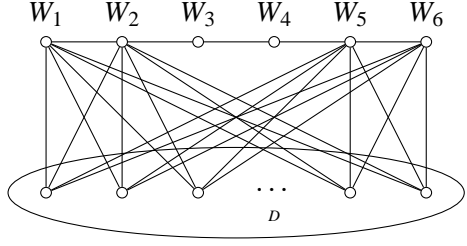

We use Figure \ref{figure} to illustrate $p$. Note that the vertices in $D$ are symmetric and $|C|=O(1)$, thus $|\{\sigma(p)| \sigma(C)=C$ and $\sigma(B)=B\}|=O(1)$. So we only need to consider the set $\{\sigma(p)|\sigma(C)\neq C$ and $\sigma(B)=B \}$,
and we exclude each $\sigma$ case by case:
\begin{enumerate}
\item  $\sigma(W_3)$ or $\sigma(W_4)$ $\in \{W_1,W_2,W_5,W_6\}$.

W.L.O.G, assume $\sigma(W_3)=W_1$, then
\begin{displaymath}
\begin{array}{lll}
dist(p,\sigma(p))&\geq&|\{e\subseteq[n]|\sigma(p)(e)=1,p(e)=0,|e|=k,\{W_3,B\}\subseteq e\}|\\
&\geq&|\{e\subseteq[n]|\sigma(p)(e)=1,|e|=k,\{W_3,B\}\subseteq e\}|\\
&&-|\{e\subseteq[n]|p(e)=\{1,\star\},|e|=k,\{W_3,B\}\subseteq e\}|\\
&=&|\{e\subseteq[n]|p(e)=1,|e|=k,\{W_1,B\}\subseteq e\}|-O(1)\\
&\geq& n-O(1)\geq3.
\end{array}
\end{displaymath}
\item  $\sigma(W_3)$ or $\sigma(W_4)$ $\in D$, and $\{\sigma(W_3),\sigma(W_4)\}\bigcap\{W_1,W_2,W_5,W_6\}=\emptyset$.

W.L.O.G, assume $\sigma(W_3)\in D$, note that for any $v$, $p(B\cup W_3\cup v)\neq\star$, and $|\{v\neq W_4|p(B\cup W_3\cup v)=1\}|=1$. While $|\{v\neq W_4|\sigma(p)(B\cup W_3\cup v\}|=|\{v\neq W_4|p(B\cup \sigma(W_3)\cup \sigma(v)\}|=4$, thus
$dist(p,\sigma(p))\geq 3$.

\item  $\sigma(W_3)=W_3$ and $\sigma(W_4)=W_4$.

\begin{enumerate}
\item $\sigma(W_5)\neq W_5$ and $\sigma(W_2)\neq W_2$.

Since $p(B\cup\sigma(W_2)\cup\sigma(W_3))=p(B\cup\sigma(W_4)\cup\sigma(W_5))=0$ and $p(B\cup\sigma(W_3)\cup\sigma(S))=p(B\cup\sigma(W_4)\cup\sigma(S'))=1$, for some $\sigma(S)=W_2$ and $\sigma(S')=W_5$, thus $dist(p,\sigma(p))\geq 4$.

\item $\sigma(W_5)=W_5$ and $\sigma(W_2)=W_2$.

Since $\sigma(C)\neq C$, W.L.O.G, assume $\sigma(W_1)\in D$, then $p(B\cup\sigma(W_1)\cup\sigma(W_5))=1$.

If $\sigma(W_6)\in D$, then $p(B\cup\sigma(W_2)\cup\sigma(W_6))=1$ and $p(B\cup\sigma(S)\cup\sigma(W_5))=p(B\cup\sigma(S')\cup\sigma(W_2))=0$, for some $\sigma(S)=W_1$ and $\sigma(S')=W_6$.

If $\sigma(W_6)=W_6$, then $p(B\cup\sigma(W_1)\cup\sigma(W_6))=1$, and $p(B\cup\sigma(S)\cup\sigma(W_5))=p(B\cup\sigma(S)\cup\sigma(W_6))=0$, for some $\sigma(S)=W_1$.

If $\sigma(W_6)=W_1$, then $p(B\cup\sigma(W_2)\cup\sigma(W_6))=1$ and $p(B\cup\sigma(W_6)\cup\sigma(W_5))=0$.

Thus we always have $dist(p,\sigma(p))\geq 3$.

\item $\sigma(W_5)\neq W_5$ and $\sigma(W_2)=W_2$.

Note that $p(B\cup\sigma(W_4)\cup\sigma(W_5))=0$ and $p(B\cup\sigma(W_4)\cup\sigma(S))=1$ for some $\sigma(S)=W_5$.

If $\sigma(W_5)\in D\cup\{W_1\}$ , then $p(B\cup\sigma(W_2)\cup\sigma(W_5))=1$.

If $\sigma(W_5)=W_6$ and $\sigma(W_6)\in D\cup\{W_1\}$, then $p(B\cup\sigma(W_2)\cup\sigma(W_6))=1$.

If $\sigma(W_5)=W_6$ and $\sigma(W_6)\in W_5$, since  $\sigma(C)\neq C$, thus $\sigma(W_1)\in D$ and $p(B\cup\sigma(W_1)\cup\sigma(W_5))=1$.

Therefore we always have $dist(p,\sigma(p))\geq 3$.
\item $\sigma(W_2)\neq W_2$ and $\sigma(W_5)=W_5$.

Similar to the above one.
\end{enumerate}
\item  $\sigma(W_3)=W_4$ and $\sigma(W_4)=W_3$.

Similar to the case where $\sigma(W_3)=W_3$ and $\sigma(W_4)=W_4$.
\end{enumerate}
\subparagraph{Case for $k_2\geq 2$}
Similarly, since $|\{\sigma(p)|\sigma(C)=C$ and $\sigma(B)=B\}|=O(1)$, we only need to consider the set $\{\sigma(p)| \sigma(C)\neq C$ and $\sigma(B)=B\}$,
 and we exclude each $\sigma$ case by case:
 \begin{enumerate}
 \item  $\sigma(W_3)$ or $\sigma(W_4)$ $\notin \{W_3,W_4\}$.

Assume $\sigma(W_3)\notin \{W_3,W_4\}$, note that for any $S\cap(B\cup W_3)=\emptyset$, $p(B\cup W_3\cup S)\neq \star$ and there are only two such $S$s to make $p=1$. While no matter what $\sigma(W_3)$ is, it's easy to see there are at least five (actually many) such $S$s to make $p(B\cup \sigma(W_3)\cup \sigma(S))=1$, thus $dist(p,\sigma(p))\geq 3$.
\item  $\sigma(W_3), \sigma(W_4)\in \{W_3, W_4\}$

W.L.O.G, assume $\sigma(W_3)=W_3$ and $\sigma(W_4)=W_4$.
\begin{enumerate}
\item $\sigma(W_5)\neq W_5$ and $\sigma(W_2)\neq W_2$.

Now $p(B\cup\sigma(W_3)\cup\sigma(W_2))=p(B\cup\sigma(W_4)\cup\sigma(W_5))=0$ and $p(B\cup\sigma(W_3)\cup\sigma(S))=p(B\cup\sigma(W_4)\cup\sigma(S'))=1$, for some $\sigma(S)=W_2$ and $\sigma(S')=W_5$. Therefore, $dist(p,\sigma(p))\geq 4$.
\item $\sigma(W_5)=W_5$ or $\sigma(W_2)=W_2$.

Assume $\sigma(W_5)=W_5$, since $\sigma(C)\cap D\neq \emptyset$, there exists some $W\in\{W_1,W_2,W_6\}$ such that $\sigma(W)\cap D\neq\emptyset$.

Moreover, for any $S\subseteq W_3\cup W_4\cup W_5$ and $S\notin\{W_3,W_4,W_5\}$ with $|S|=k_2$,  we have $\sigma(S)\subseteq W_3\cup W_4\cup W_5$ and $\sigma(S)\notin\{W_3,W_4,W_5\}$ , thus $p(B\cup W\cup S)=0\neq p(B\cup\sigma(W)\cup\sigma(S))=1$, and note that there are at least ${6 \choose 2}-3=12$ such $S$s. Thus,
$dist(p,\sigma(p))\geq 3$.
\end{enumerate}
\end{enumerate}
\end{proof}

\begin{proof} ({\bf Proof of Theorem \ref{upperbound:3k+1}}) We still use minterm functions here.

Let $k=3l+1$. Note that in the above construction for $(3l+1)$-uniform hypergraph properties, $s_1(f)\leq |p|=O(n^l)$ and $s_0(f)=O(n^{l+1})$. Intuitively, we can pack $\sqrt{n}$ minterms together to get a super minterm, expecting to decrease the number of isomorphisms satisfying the distance condition (i.e., where any of two isomorphisms of $p$ have distance at most 2). Unfortunately, just packing minterms naively doesn't work here, we need some tricks.

Let $p$ be the minterm defining $f$. $p$ is constructed as follow:

The notions $V$, $B$, $W_i$, $C$ and $D$ are defined the same as in Theorem \ref{upperbound:general}, where we let $k_1=k_2=l$. Besides that, let $D_1=\{v_{6l+1},v_{6l+2},\cdots,v_{6l+\sqrt{n}}\}$ and $D_2=D\setminus D_1$.
\begin{itemize}
\item For any $S\subseteq C$ of size $2l$ and any $v\in D_1$, $p(B\cup S\cup v)=0$, except $S=W_i\cup W_{i+1}$ for $i\in[5]$ where $p(B\cup S\cup v)=1$.
\item For any $S\subseteq C$ of size $2l$ and any $v\in D_2$, $p(B\cup S\cup v)=1$.
\item For any $S$ satisfying $l\leq|S\cap C|<2l$, $|S|=2l+1$ and $S\cap D_1\neq\emptyset$, $p(B\cup S)=1$, except $W_3$ or $W_4\subseteq S$ where $p(B\cup S)=0$.
\item  All the other variables are  $\star$.
\end{itemize}
It is not hard to see that $|p|=O(n^{l+1/2})$, thus $s_1(f)\leq C_1(f)\leq|p|=O(n^{l+1/2})$.


Similar to the argument in the proof of Theorem~\ref{upperbound:general},  we just need to show the following claim to complete the proof.
 \begin{Claim}
There are only $O(\sqrt{n})$ $\sigma(p)$s  with the same $\pi(B)=\sigma(B)$ satisfying $dist(\pi(p),\sigma(p))\leq 2$.
\end{Claim}
\noindent\emph{Proof.}
By contradiction, suppose there are $C\sqrt{n}$ such $\sigma(p)$s where $C$ is a sufficient large number, thus there must exist a vertice $v$ such that $\sigma(v)\in D_1$  for at least $C$ such $\sigma(p)$s, w.l.o.g, assume this set contains $p$. And we will argue that there are only $O(1)$ such $\sigma(p)$s satisfying $dist(\sigma(p),p)\leq 2$, then it's a contradiction, which completes the proof.

Since the vertices in $D_1$ or $D_2$ are symmetric, thus $|\{\sigma(p)|\sigma(C)=C$ and $\sigma(D_1)=D_1\}|=O(1)$.

If $\sigma(C)=C$ and $\exists v_1\in D_1,v_2\in D_2$ satisfying $\sigma(v_1)=v_2$, then $dist(\sigma(p),p)\geq3$, since almost all variables which contains $v_1$, $C$ and $B$ are 0 in $p$, while all these variable are 1 in $\sigma(p)$.

If $\sigma(C)\neq C$, since $\sigma(v)\in D_1$, then we find that $p(S\cup v)=p'(S)$ where $p'$ is the minterm defined in Theorem~\ref{upperbound:general} for $3l$-uniform hypergraph properties. Similarly, $\sigma(p)(S\cup v)=p(\sigma(S)\cup \sigma(v))=p'(\sigma(S))$. We only consider those $S$s satisfying $v\notin \sigma(S)\cup S$ and follows the similar proof of Claim 1 in Theorem \ref{upperbound:general}. Finally we can obtain $dist(p,\sigma(p))\geq3$.

\end{proof}

\section{$k$-Partite $k$-Uniform Hypergraph Properties and Abelian Groups}\label{section:general}
In this section, we give the constructions of $k$-partite $k$-uniform hypergraph properties first.
\begin{proof} ({\bf Proof of Theorem \ref{thm:partite1}})

\begin{table}[!hbp]
\begin{center}
\begin{tabular}{|c|c|c|c|c|c|c|c|c|}
  \hline
  $\vec{b}=$ & $\vec{1}$ & $\vec{2}$ & $\vec{3}$ & $\vec{4}$ & $\vec{5}$ & $\vec{6}$ & $\vec{7}$ & $\cdots$ \\ \hline
  $\vec{a}=\vec{1}$ & 0 & 0 & 0 & 1 & 1 & 1 & 1 & 1 \\
  $\vec{2}$ & 0 & 0 & 1 & 1 & 1 & 1 & 1 & 1 \\
  $\vec{3}$ & 1 & 0 & 1 & 0 & 0 & 0 & 0 & 0 \\
  $\vec{4}$ & 1 & 1 & 0 & $\star$ & $\star$ & $\star$ & $\star$ & $\star$ \\
  $\vec{5}$ & 1 & 1 & 0 & $\star$ & $\star$ & $\star$ & $\star$ & $\star$ \\
  $\vec{6}$ & 1 & 1 & 0 & $\star$ & $\star$ & $\star$ & $\star$ & $\star$ \\
  $\vec{7}$ & 1 & 1 & 0 & $\star$ & $\star$ & $\star$ & $\star$ & $\star$ \\
  $\cdots$ & 1 & 1 & 0 & $\star$ & $\star$ & $\star$ & $\star$ & $\star$ \\
  \hline
\end{tabular}
\caption{The tabel to illustrate $p$ of $k$-partite $k$-uniform hypergraph properties.}\label{table}
\end{center}
\end{table}

%
The function we use here is  also a minterm function. Let $k_1$ and $k_2$ be the integers such that $k_1+2k_2=k$ and $k_1,k_2\leq \lceil k/3\rceil$. We divide the $k$ partitions into three sets, and each of them is of size $k_2$, $k_2$ and $k_1$ and indicated by $\vec{a},\vec{b}\in[n]^{k_2}$ and $\vec{c}\in[n]^{k_1}$  respectively.
We use Table~\ref{table} to illustrate the minterm $p$:
\begin{itemize}
\item For $\vec{b}=\vec{1},\vec{2}$ or $\vec{3}$, $p{(\vec{1},\vec{b},\vec{1}_c)}=0$, otherwise $p{(\vec{1},\vec{b},\vec{1}_c)}=1$.
\item For $\vec{b}=\vec{1}$ or $\vec{2}$, $p{(\vec{2},\vec{b},\vec{1}_c)}=0$, otherwise $p{(\vec{2},\vec{b},\vec{1}_c)}=1$.
\item For $\vec{b}=\vec{1}$ or $\vec{3}$, $p{(\vec{3},\vec{b},\vec{1}_c)}=1$, otherwise $p{(\vec{3},\vec{b},\vec{1}_c)}=0$.
\item For $\vec{a}\notin\{\vec{1},\vec{2},\vec{3}\}$ and $\vec{b}=\vec{1}$ or $\vec{2}$, $p{(\vec{a},\vec{b},\vec{1}_c)}=1$.
\item For $\vec{a}\notin\{\vec{1},\vec{2},\vec{3}\}$ and $\vec{b}=\vec{3}$, $p{(\vec{a},\vec{b},\vec{1}_c)}=0$.
\item Otherwise $p{(\vec{a},\vec{b},\vec{c})}=\star$.
\end{itemize}
Here $\vec{1}$, $\vec{2}$ and $\vec{3}$ can be any three different vectors, W.L.O.G, assume $\vec{1}=(1,\cdots,1,1)$, $\vec{2}=(1,\cdots,1,2)$, $\vec{3}=(1,\cdots,1,3)$ and $\vec{1}_c=(1,\cdots,1,1)$.
It's easy to see $s_1(f)\leq C_1(f)\leq|p|=O(n^{k_2})$.
By discussing case by case, it can be verified that for any $p^{\pi}$ there are at most $O(1)$ $p^{\sigma}$s satisfying $\pi(\vec{c})=\sigma(\vec{c})$ and $dist(p^{\pi},p^{\sigma})\leq2$. Thus $s_0(f)=O(n^{k_1})$ since there are at most $n^{k_1}$ choices of $\vec{c}$. The verify procedure is trivial but tedious, and we omit it here.
\end{proof}

In the following, we give the proofs of Theorem~\ref{thm:partite2} and Theorem~\ref{thm:abelian}.
\begin{proof} ({\bf Proof of Theorem~\ref{thm:partite2}})
We still use minterm functions here. Let $k=3l+1$ where $l\geq1$. We divide the $k$ partitions into four sets of size $l$, $l$, $l$ and 1, and each set is indicated by $\vec{a},\vec{b},\vec{c}\in[n]^l$ and $\vec{d}\in[n]$ respectively.
The minterm $p$ is constructed as follow:
\begin{itemize}
\item For any $\vec{d}\in[\sqrt{n}]$, and any $\vec{a}$ and $\vec{b}$, $p(\vec{a},\vec{b},\vec{1},\vec{d})=p'(\vec{a},\vec{b},\vec{1}_c)$. Here $p'$ is the partial assignment defined in the proof of Theorem~\ref{thm:partite1}.
\item For any $\vec{d}\notin[\sqrt{n}]$ and $\vec{a},\vec{b}\in\{\vec{1},\vec{2},\vec{3}\}$, $p(\vec{a},\vec{b},\vec{1},\vec{d})=1$.
\item Otherwise $p(\vec{a},\vec{b},\vec{c},\vec{d})=\star$.
\end{itemize}
It's easy to see $s_1(f)\leq|p|=O(n^{l+1/2})$.
It is also not hard to verify that there are at most $\sqrt{n}$ $p^{\sigma}$s with the same $\sigma({\vec{c}})$ and satisfying the condition that the distance between any two of them is at most 2, thus $s_0(f)=O(n^{l+1/2})$.
\end{proof}
\begin{proof} ({\bf Proof of Theorem~\ref{thm:abelian}})
First note that the transitive action of a group $G$ on $[n]$ is equivalent to the action of $G$ by left multiplication on a coset space $G/$Stab$_1$, here Stab$_1$ is the stabilizer of  the element $1\in[n]$. Since $G$ is an Abelian group, Stab$_1$=$\cdots$=Stab$_n$, thus Stab$_1$=$\{e\}$. Therefore, the action of $G$ on $[n]$ is equivalent to the action of $G$ by multiplication on itself. So we can relabel the variables $(x_1,\cdots,x_n)$ as $(x_{(1,\cdots,1)},\cdots,x_{(m_1,\cdots,m_l)})$ to make $(\sigma_1\otimes\cdots\otimes\sigma_l)(x)=(x_{(\sigma_1(1), \cdots,\sigma_l(1))},\cdots,x_{(\sigma(m_1),\cdots,\sigma_l(m_2))})$ for any $\sigma_1\otimes\cdots\otimes\sigma_l\in C_{m_1}\times\cdots\times C_{m_l}$.

Let $p_m$ be the minterm of $f:\{0,1\}^m\rightarrow\{0,1\}$ defined by Chakraborty in Theorem 3.1 in~\cite{Chakraborty}.
We define the minterm $p$ as $p(i_1,\cdots,i_l)=\bigoplus_{j=1}^{l} p_{m_j}(i_j)$. Here $\star\oplus b=\star$, for $b=0,1$, or $\star$.
It is easy to see $s_1(f)\leq|p|=\prod_{j=1}^{l}|p_{m_j}|\leq\gamma n^{1/3}$, where $\gamma$ is a number only depending on $l$.
Moreover, according to the construction of $p_m$, it is easy to see that there are at most $\beta n^{1/3}$ $\sigma(p)$s  where the distance between any two of them is at most 2. Here $\beta$ is another number only depending on $l$, thus $s_0(f)\leq \beta n^{1/3}$. This completes the proof.
\end{proof}

\section{Lower bounds}\label{section:proof_2}
In this section, we give the proof of Theorem~\ref{lowerbound}. The proof is similar with Lemma 8 in \cite{Sun}.

\begin{proof} ({\bf Proof of Theorem~\ref{lowerbound}}) W.O.L.G we assume that for the empty graph $\overline{K_n}$, $f(\overline K_n)=0$. Since $f$ is non-trivial, there must exist a graph $G$ such that $f(G)=1$. Let's consider graphs in $f^{-1}(1)=\{G|f(G)=1\}$ with the minimum number of edges. Define $m=\min\{|E(G)|:f(G)=1\}$.

We claim that if $m\geq\frac{1}{k+2}n$, then $s(f)\geq\frac{1}{k+2}n$. Let $G$ be a graph in $f^{-1}(1)$ and $|E(G)|=m$. Consider the subfunction $f'$ where $\forall e\notin E(G)$, $x_e$ is restricted to 0, since G has the the minimum number of edges, deleting any edges from $G$ will change the values of $f(G)$, therefore, $f'$ is a AND function. Thus, $s(f)\geq s(f')=m\geq\frac{1}{k+2}n$.

In the following we assume $m<\frac{1}{k+2}n$. Again let $G$ be a graph in $f^{-1}(1)$ with $|E(G)|=m$. Let us consider the isolated vertices set $I$, as
$$\sum_{v\in V}deg(v)=k|E(G)|<\frac{k}{k+2}n.$$
We have
$$|I|\geq n-\sum_{v\in V}deg(v)>\frac{2}{k+2}n.$$
Suppose $s(f)<\frac{1}{k+2}n$, we will deduce that there exists another graph  with fewer edges and the same value, against the assumption that $G$ has the minimum number of edges in $f^{-1}(1)$, which ends the whole proof.

 Pick a vertex $u$ with $deg(u)=d>0$. Suppose in the graph $G$ vertex $u$ is adjacent to $(k-1)$-edges $\{e^{(k-1)}_1,e^{(k-1)}_2,\cdots,e^{(k-1)}_d\}$ and  $I=\{u_1,u_2,\cdots,u_t\}$, where $t=|I|$.

Consider the $t$-variable Boolean function $g_1$: $\{0,1\}^t\rightarrow \{0,1\}$, where
$$g_1(x_1,\cdots,x_t)=f(G+x_1(e^{(k-1)}_1,u_1)+\cdots+x_t(e_1^{(k-1)},u_t)).$$
It is easy to see that $g_1$ is a symmetric function. We claim that  $g_1$ is a constant function: if not, we have $s(g_1)\geq\frac{1}{2}t$  \cite{turan}, which implies $s(f)>\frac{1}{k+2}n$ since $g_1$ is a restriction of $f$. In particular, $g_1(1,\cdots,1)=g_1(0,\cdots,0)$, i.e. $f(G_1)=f(G)$, where $G_1=G+\sum_{i=1}^t(e^{(k-1)}_1,u_i)$.

Define $G_i=G_{i-1}+\sum_{j=1}^t(e^{(k-1)}_i,u_j)$ $(i=2,\cdots,d)$. Similarly, we can show that
$$f(G)=f(G_1)=\cdots=f(G_d).$$
Next we will delete all the edges between $\{u,u_1,\cdots,u_t\}$ and $\{e^{(k-1)}_1,e^{(k-1)}_2,\cdots,e^{(k-1)}_d\}$ from $G_d$ by reversing the adding edge procedure of $G\rightarrow G_1\rightarrow\cdots \rightarrow G_d$. More precisely, define $H_1=G_d$; for $i=2,\cdots,d$, define
$$H_i=H_{i-1}-(e^{(k-1)}_i,u)-(e^{(k-1)}_i,u_1)-\cdots-(e^{(k-1)}_i,u_t),$$
and
$$h_i(y_0,y_1,\cdots,y_t)=f(H_i+y_0(e^{(k-1)}_i,u)+y_1(e^{(k-1)}_i,u_1)+\cdots+y_t(e^{(k-1)}_i,u_t)).$$
Similarly, by the fact $s(f)<\frac{1}{k+2}n$ we can show that all the functions $h_2,\cdots,h_d$ are constant, which implies $f(H_1)=f(H_2)=\cdots=f(H_d)$. So we find another graph $H_d$ with fewer edges than $G$ and $f(H_d)=1$.
\end{proof}

\section{Conclusion}\label{section:conclusion}
In this paper, we present a $k$-uniform hypergraph property with sensitivity complexity $O(n^{\lceil k/3\rceil})$ for any $k\geq 3$ and we can do better when $k\equiv1$ (mod 3). Besides that, we also investigate the sensitivity complexity of other transitive Boolean functions with certain symmetry.  All the functions we constructed in this paper are minterm transitive functions. On the other side,
Charkrobati~\cite{Chakraborty} proved that the sensitivity complexity of any minterm transitive Boolean function $f:\{0,1\}^n\rightarrow \{0,1\}$ is at least $\Omega(n^{1/3})$.
Kulkarni et al.~\cite{fractional} point out that the existence of any transitive function $f:\{0,1\}^n\rightarrow \{0,1\}$ with $s(f)=n^\alpha$ where $\alpha<1/3$ implies a  larger than quadratic separation between block sensitivity and sensitivity.  We conjecture that the example here is almost tight.
\begin{Conjecture}\label{conj2}
For any constant $k\geq 3$ and for any non-trivial $k$-hypergraph property $f$,  $s(f)=\Omega(n^{k/3})$, where $n$ is the number of vertices.
\end{Conjecture}
\begin{Conjecture}\label{conj1}
For any $k\geq 3$, there exist a sequence of $k$-uniform hypergraph properties $f$ with $s(f)=O(n^{ k/3 })$, where $n$ is the number of vertices.
\end{Conjecture}

A more general question is the following variant of Turan's question proposed by Chakraborty~\cite{Chakraborty}:
If $f$ is Boolean function invariant under a transitive group of permutations then is it true that $s(f)= \Omega(n^c)$ for some constant $c>0$? We conjecture that the inequality holds for $c=1/3$, which would imply Conjecture~\ref{conj2} and the sensitivity conjecture of transitive functions.

\bibliographystyle{unsrt}
\bibliography{reference}

\begin{thebibliography}{10}

\bibitem{Babai}
Joshua Biderman, Kevin Cuddy, Ang Li, and Min~Jae Song.
\newblock On the sensitivity of k-uniform hypergraph properties.
\newblock {\em CoRR}, abs/1510.00354, 2015.

\bibitem{Evasiveness}
L{\'{a}}szl{\'{o}} Lov{\'{a}}sz and Neal~E. Young.
\newblock Lecture notes on evasiveness of graph properties.
\newblock {\em CoRR}, cs.CC/0205031, 2002.

\bibitem{CKS01}
Amit Chakrabarti, Subhash Khot, and Yaoyun Shi.
\newblock Evasiveness of subgraph containment and related properties.
\newblock {\em {SIAM} J. Comput.}, 31(3):866--875, 2001.

\bibitem{Kulkarni13}
Raghav Kulkarni.
\newblock Evasiveness through a circuit lens.
\newblock In {\em Innovations in Theoretical Computer Science, {ITCS} '13,
  Berkeley, CA, USA, January 9-12, 2013}, pages 139--144, 2013.

\bibitem{Stacs10}
L{\'{a}}szl{\'{o}} Babai, Anandam Banerjee, Raghav Kulkarni, and Vipul Naik.
\newblock Evasiveness and the distribution of prime numbers.
\newblock In {\em 27th International Symposium on Theoretical Aspects of
  Computer Science, {STACS} 2010, March 4-6, 2010, Nancy, France}, pages
  71--82, 2010.

\bibitem{RivestV76}
Ronald~L. Rivest and Jean Vuillemin.
\newblock On recognizing graph properties from adjacency matrices.
\newblock {\em Theor. Comput. Sci.}, 3(3):371--384, 1976.

\bibitem{KQS15}
Raghav Kulkarni, Youming Qiao, and Xiaoming Sun.
\newblock Any monotone property of 3-uniform hypergraphs is weakly evasive.
\newblock {\em Theor. Comput. Sci.}, 588:16--23, 2015.

\bibitem{Black2015}
Timothy Black.
\newblock Monotone properties of k-uniform hypergraphs are weakly evasive.
\newblock In {\em Proceedings of the 2015 Conference on Innovations in
  Theoretical Computer Science}, ITCS '15, pages 383--391, New York, NY, USA,
  2015. ACM.

\bibitem{Sun07}
Xiaoming Sun.
\newblock Block sensitivity of weakly symmetric functions.
\newblock {\em Theor. Comput. Sci.}, 384(1):87--91, 2007.

\bibitem{turan}
Gy{\"o}rgy Tur{\'a}n.
\newblock The critical complexity of graph properties.
\newblock {\em Information Processing Letters}, 18(3):151--153, 1984.

\bibitem{Sun}
Xiaoming Sun.
\newblock An improved lower bound on the sensitivity complexity of graph
  properties.
\newblock {\em Theoretical Computer Science}, 412(29):3524 -- 3529, 2011.

\bibitem{GMSZ13}
Yihan Gao, Jieming Mao, Xiaoming Sun, and Song Zuo.
\newblock On the sensitivity complexity of bipartite graph properties.
\newblock {\em Theor. Comput. Sci.}, 468:83--91, 2013.

\bibitem{Chakraborty}
Sourav Chakraborty.
\newblock On the sensitivity of cyclically-invariant boolean functions.
\newblock In {\em Proceedings of the 20th Annual IEEE Conference on
  Computational Complexity}, CCC '05, pages 163--167, Washington, DC, USA,
  2005. IEEE Computer Society.

\bibitem{Cook}
Stephen Cook and Cynthia Dwork.
\newblock Bounds on the time for parallel ram's to compute simple functions.
\newblock In {\em Proceedings of the Fourteenth Annual ACM Symposium on Theory
  of Computing}, STOC '82, pages 231--233, New York, NY, USA, 1982. ACM.

\bibitem{CDR86}
Stephen~A. Cook, Cynthia Dwork, and R{\"{u}}diger Reischuk.
\newblock Upper and lower time bounds for parallel random access machines
  without simultaneous writes.
\newblock {\em {SIAM} J. Comput.}, 15(1):87--97, 1986.

\bibitem{nisan1991crew}
Noam Nisan.
\newblock Crew prams and decision trees.
\newblock {\em SIAM Journal on Computing}, 20(6):999--1007, 1991.

\bibitem{buhrman}
Harry Buhrman and Ronald De~Wolf.
\newblock Complexity measures and decision tree complexity: a survey.
\newblock {\em Theoretical Computer Science}, 288(1):21--43, 2002.

\bibitem{NS92}
Noam Nisan and Mario Szegedy.
\newblock On the degree of boolean functions as real polynomials.
\newblock In {\em Proceedings of the Twenty-fourth Annual ACM Symposium on
  Theory of Computing}, STOC '92, pages 462--467, New York, NY, USA, 1992. ACM.

\bibitem{APV16}
Andris Ambainis, Krisjanis Prusis, and Jevgenijs Vihrovs.
\newblock Sensitivity versus certificate complexity of boolean functions.
\newblock In {\em Computer Science - Theory and Applications - 11th
  International Computer Science Symposium in Russia, {CSR} 2016, St.
  Petersburg, Russia, June 9-13, 2016, Proceedings}, pages 16--28, 2016.

\bibitem{Own}
Kun He, Qian Li, and Xiaoming Sun.
\newblock A tighter relation between sensitivity and certificate complexity.
\newblock {\em Manuscript}.

\bibitem{AS11}
Andris Ambainis and Xiaoming Sun.
\newblock New separation between $s(f)$ and $bs(f)$.
\newblock {\em Electronic Colloquium on Computational Complexity {(ECCC)}},
  18:116, 2011.

\bibitem{HKP11}
Pooya Hatami, Raghav Kulkarni, and Denis Pankratov.
\newblock Variations on the sensitivity conjecture.
\newblock {\em Theory of Computing, Graduate Surveys}, 4:1--27, 2011.

\bibitem{Bop12}
Meena Boppana.
\newblock Lattice variant of the sensitivity conjecture.
\newblock {\em Electronic Colloquium on Computational Complexity {(ECCC)}},
  19:89, 2012.

\bibitem{AP14}
Andris Ambainis and Krisjanis Prusis.
\newblock A tight lower bound on certificate complexity in terms of block
  sensitivity and sensitivity.
\newblock In {\em Mathematical Foundations of Computer Science 2014 - 39th
  International Symposium, {MFCS} 2014, Budapest, Hungary, August 25-29, 2014.
  Proceedings, Part {II}}, pages 33--44, 2014.

\bibitem{ABG+14}
Andris Ambainis, Mohammad Bavarian, Yihan Gao, Jieming Mao, Xiaoming Sun, and
  Song Zuo.
\newblock Tighter relations between sensitivity and other complexity measures.
\newblock In {\em Automata, Languages, and Programming - 41st International
  Colloquium, {ICALP} 2014, Copenhagen, Denmark, July 8-11, 2014, Proceedings,
  Part {I}}, pages 101--113, 2014.

\bibitem{AV15}
Andris Ambainis and Jevgenijs Vihrovs.
\newblock Size of sets with small sensitivity: {A} generalization of simon's
  lemma.
\newblock In {\em Theory and Applications of Models of Computation - 12th
  Annual Conference, {TAMC} 2015, Singapore, May 18-20, 2015, Proceedings},
  pages 122--133, 2015.

\bibitem{GKS15}
Justin Gilmer, Michal Kouck{\'{y}}, and Michael~E. Saks.
\newblock A new approach to the sensitivity conjecture.
\newblock In {\em Proceedings of the 2015 Conference on Innovations in
  Theoretical Computer Science, {ITCS} 2015, Rehovot, Israel, January 11-13,
  2015}, pages 247--254, 2015.

\bibitem{Sze15}
Mario Szegedy.
\newblock An $o(n^{0.4732})$ upper bound on the complexity of the {GKS}
  communication game.
\newblock {\em Electronic Colloquium on Computational Complexity {(ECCC)}},
  22:102, 2015.

\bibitem{GNS+16}
Parikshit Gopalan, Noam Nisan, Rocco~A. Servedio, Kunal Talwar, and Avi
  Wigderson.
\newblock Smooth boolean functions are easy: Efficient algorithms for
  low-sensitivity functions.
\newblock In {\em Proceedings of the 2016 {ACM} Conference on Innovations in
  Theoretical Computer Science, Cambridge, MA, USA, January 14-16, 2016}, pages
  59--70, 2016.

\bibitem{GSTW16}
Parikshit Gopalan, Rocco~A. Servedio, and Avi Wigderson.
\newblock Degree and sensitivity: Tails of two distributions.
\newblock In {\em 31st Conference on Computational Complexity, {CCC} 2016, May
  29 to June 1, 2016, Tokyo, Japan}, pages 13:1--13:23, 2016.

\bibitem{Tal16}
Avishay Tal.
\newblock On the sensitivity conjecture.
\newblock {\em Electronic Colloquium on Computational Complexity {(ECCC)}},
  23:62, 2016.

\bibitem{Sha16}
Shalev Ben{-}David.
\newblock Low-sensitivity functions from unambiguous certificates.
\newblock {\em Electronic Colloquium on Computational Complexity {(ECCC)}},
  23:84, 2016.

\bibitem{LZ16}
Chengyu Lin and Shengyu Zhang.
\newblock Sensitivity conjecture and log-rank conjecture for functions with
  small alternating numbers.
\newblock {\em CoRR}, abs/1602.06627, 2016.

\bibitem{ST16}
Karthik~C. S. and S{\'{e}}bastien Tavenas.
\newblock On the sensitivity conjecture for disjunctive normal forms.
\newblock {\em CoRR}, abs/1607.05189, 2016.

\bibitem{fractional}
Raghav Kulkarni and Avishay Tal.
\newblock On fractional block sensitivity.
\newblock {\em Chicago J. Theor. Comput. Sci.}, 2016, 2016.

\end{thebibliography}
\end{document}